\documentclass[11pt]{article}
\usepackage{fullpage}
\usepackage{latexsym}
\usepackage{amsmath}
\usepackage{amssymb}
\usepackage{amsthm}
\usepackage{hyperref}
\usepackage{graphicx}
\usepackage[english]{babel}
\usepackage{color}
\usepackage{amstext}
\usepackage{array}
\usepackage{xspace}
\usepackage{calc}

\usepackage{tikz}
\usetikzlibrary{decorations.markings}
\tikzstyle{vertex}=[circle, draw, inner sep=0pt, minimum size=6pt]

\newtheorem{theorem}{Theorem}[section]
\newtheorem{theorem*}{Theorem}
\newtheorem{corollary}[theorem]{Corollary}
\newtheorem{lemma}[theorem]{Lemma}

\newtheorem{proposition}[theorem]{Proposition}

\newtheorem{fact}[theorem]{Fact}

\theoremstyle{definition}
\newtheorem{definition}[theorem]{Definition}

\theoremstyle{remark}
\newtheorem{remark}[theorem]{Remark}

\numberwithin{equation}{section}

\renewcommand{\vec}[1]{\mathbf{#1}}
\newcommand{\SL}{\mathrm{SL}}

\newcommand{\F}{\mathbb{F}}
\newcommand{\K}{\mathbb{K}}
\newcommand{\ELL}{\mathbb{L}}
\newcommand{\Z}{\mathbb{Z}}
\newcommand{\Q}{\mathbb{Q}}
\newcommand{\C}{\mathbb{C}}
\newcommand{\N}{\mathbb{N}}

\newcommand{\Tr}{\mathrm{Tr}}

\newcommand{\NP}{\mathrm{NP}}

\newcommand{\cB}{\mathcal{B}}
\newcommand{\cA}{\mathcal{A}}
\newcommand{\rk}{\mathrm{rk}}
\newcommand{\crk}{\mathrm{crk}}
\newcommand{\ncrk}{\mathrm{ncrk}}
\newcommand{\cork}{\mathrm{cork}}
\newcommand{\blowup}[2]{{#1}^{[#2]}}
\newcommand{\rblowup}[2]{{#1}^{\{#2\}}}

\newcommand{\trans}[1]{{#1}^{\mathrm{T}}}
\newcommand{\fdchar}{\mathrm{char}}

\newcommand{\nrk}{\mathrm{ncrk}}
\newcommand{\poly}{\mathrm{poly}}
\newcommand{\im}{\mathrm{im}}
\newcommand{\zvec}{\mathbf{0}}

\newcommand{\linspan}{\mathrm{span}}
\title{Constructive non-commutative rank computation is in deterministic                    
polynomial time}% Insert title                                                                     
\author{
G\'abor Ivanyos \thanks{ Institute for Computer Science and Control, Hungarian
Academy of Sciences , Budapest, Hungary. {E-mail: \tt Gabor.Ivanyos@sztaki.mta.hu}} 
\and
Youming Qiao  \thanks{Centre for Quantum Software and Information,
University of Technology Sydney, Sydney, Australia.
{E-mail: \tt jimmyqiao86@gmail.com} }\and
K. V. Subrahmanyam \thanks{ Chennai Mathematical Institute, Chennai, India.
{E-mail: \tt kv@cmi.ac.in}}
}

\begin{document}
\maketitle

\begin{abstract}
We extend the  techniques developed in \cite{IQS1}
%with methods from a recent paper by Derksen and Makam
to obtain a deterministic polynomial-time algorithm for
computing the non-commutative rank
of linear spaces of matrices over any field.

The key new idea that causes a reduction in the time complexity of the algorithm in \cite{IQS1} from exponential time to polynomial time is a reduction procedure that keeps the blow-up 
parameter small,  and there are two methods to implement this idea: the first one is a 
greedy argument that removes certain rows and columns, and the second one is an 
efficient  algorithmic version of a result of Derksen and Makam \cite{DM2}, who were the first to observe that the blow-up parameter can be controlled.
Both methods rely crucially 
on the regularity lemma from \cite{IQS1}. In this note we 
improve that lemma by removing a coprime condition there. 
%The result can be considered as an efficient algorithmic version
%a result of Derksen and Makam.
\end{abstract}

%\begin{keywords}
%Edmonds' problem, symbolic
%determinant identity test, semi-invariants of quivers, non-commutative rank
%\end{keywords}

%\begin{subject}
%13A50, 68W30
%\end{subject}

%\begin{document}

\section{Introduction}\label{sec:intro}
%: blow-ups and the non-commutative rank}

This paper builds on the work reported in our previous paper \cite{IQS1}. In the interest of keeping this paper self contained we introduce the problem again, recall 
its connections to invariant theory and operator theory, and describe recent progress on this problem including our work, \cite{IQS1}, the work of
Garg, Gurvits, Oliviera and Wigderson \cite{GGOW}, and that of Derksen and Makam 
\cite{DM2}. As a result this introduction 
overlaps with the introduction in \cite{IQS1}. Readers who are familiar with \cite{IQS1} can skip straight to ~\ref{sec:progress} where we describe the 
new results in this paper.

Let $X=\{x_1, \dots, x_m\}$ be a set of variables. Given an
$n\times n$ matrix $T$ whose entries are homogeneous linear polynomials from $\Z[X]$, determining the
rank of $T$ over the rational function field $\Q(X)$ is a fundamental open 
problem. This problem, denoted $\rk(T)$, was introduced by J. Edmonds 
\cite{Edm67}.  The decision version of this problem, deciding whether $T$ has rank 
$n$ is known as the Symbolic Determinant Identity 
Testing problem (SDIT). It is natural to consider the problem over any field $\F$. 
If $|\F|$ is constant, this problem was shown to be $\NP$-hard \cite{BFS}. This is 
not the setting we will be concerned with -- we 
will always
assume $|\F|$ to be at least $\Omega(n)$.

When $|\F|\geq 2n$, the Schwartz-Zippel lemma provides a randomized efficient
algorithm for SDIT. Devising a deterministic efficient algorithm for this problem has a long history and
is of fundamental importance in complexity theory. In 2003, Kabanets and 
Impagliazzo \cite{KI04} showed a remarkable connection between deterministic 
efficient algorithms for SDIT and 
circuit lower bounds. This endows SDIT with fundamental importance in 
computational complexity, but the problem still remains hugely open. Improving on 
the results in  \cite{KI04},
Carmosino et al \cite{CIKK15} showed that an efficient algorithm for SDIT implies the existence of an 
explicit multilinear polynomial family such that its graph is computable in $\mathrm{NE}$, but the polynomial family cannot be computed by
polynomial-size arithmetic circuits. 

It is also natural to consider this problem in the non-commutative setting. The
\emph{free skew field} is the non-commutative analogue of the rational function
field. We do not define the free skew field in this paper and only point out that the free skew field was first constructed by Amitsur \cite{Ami66}, and
alternative
constructions were given subsequently
by Bergman \cite{Berg70}, Cohn \cite{Cohn}, and Malcolmson \cite{Mal78}.  
%\ynote{Added: }
We refer the reader to \cite{HW15} by Hrube\v{s} and Wigderson for 
a nice 
introduction to the free skew field from the perspective of algebraic 
computations. 
Cohn's books \cite{Cohn,Cohn95} serve as a comprehensive 
introduction to this topic. 
%\ynote{Deleted: Cohn  
%defined the non-commutative rank of  $T$ as the
%minimum $s\in\Z^+$ s.t. $T$ can be written as $PQ$,
%where $P$ and $Q$ are matrices of size $n\times s$ and $s\times n$ with entries 
%homogeneous linear polynomials in 
%$\F\langle                                                                      
%X\rangle$ \cite{Cohn95}. Fortin and Reutenauer \cite{FR04} showed that  this is 
%the same as rank of $T$ over the
%free skew field. }
By the \emph{non-commutative Edmonds problem} we
mean the problem of computing the non-commutative rank of $T$, denoted $\nrk(T)$, and by the \emph{non-commutative                                       
full rank problem} (NCFullRank) we mean the problem of deciding whether
$\nrk(T)$ is full or not. Cohn and Reutenauer \cite{CR99} showed that  NCFullRank is in PSPACE.

In order to talk about further progress on $\nrk(T)$ and NCFullrank we need to describe the various {\em avatars} of 
the non-commutative rank. We give four equivalent formulations of 
the non-commutative rank. We do not give full proofs that these are equivalent formulations since the proofs were already sketched in ~\cite{IQS1}.  
We recall some important definitions from \cite{IQS1} needed to describe these formulations. 

First some notation. Let $M(n, \F)$ denote the
linear space of
$n\times n$ matrices over $\F$. A linear subspace of $M(n, \F)$ is called a
\emph{matrix space}. Given $T$ a matrix of linear forms in variables
$X=\{x_1, \dots, x_m\}$ write $T=x_1B_1+x_2B_2+\dots+x_mB_m$, where
$B_i\in M(n, \F)$.
%The matrix space spanned by $B_i$'s is denoted by                                                               
Let $\cB:=\langle B_1, \dots, B_m\rangle$, where $\langle \cdot \rangle$ denotes
linear span. The rank of $\cB$, denoted as $\rk(\cB)$, is defined as
$\max\{\rk(B)\mid B\in\cB\}$. We call $\cB$ \emph{singular}, if $\rk(\cB)<n$. When
$|\F| > n$, as we will assume throughout,
$\rk(T)=\rk(\cB)$; this is because when the field size is large
enough, the complement of the zero set of a nonzero polynomial is non-empty. 

%We
%shall soon see
%that $\nrk(T)$ corresponds to some property
%of $\cB$ as well, so that we can translate the study of commutative and
%non-commutative ranks of $T$ entirely to the study of $\cB$.
%\newpage
\paragraph{Shrunk subspaces:}
\begin{definition}
Given $\cB=\langle B_1, \dots, B_m\rangle \leq M(n, \F)$, a subspace $U\leq \F^n$ is called a
\emph{$c$-shrunk subspace} of $\cB$ for $c\in\N$, if there exists $W\leq \F^n$,
such that $\dim(W)\leq \dim(U)-c$ and for every $B\in\cB$, $B(U)\leq W$. $U$ is
called a shrunk subspace of $\cB$, if it is a $c$-shrunk subspace for some $c\in\Z^+$.
\end{definition}

Cohn showed that the non-commutative rank is not full if and only if 
there
is a shrunk subspace \cite{Cohn95}. This was generalized by Fortin and Reutenauer 
\cite[Theorem 1]{FR04}, where the authors showed
$$\nrk(T)=n-\max\{c\in\{0, 1, \dots, n\} \mid \exists c\text{-shrunk subspace of }\cB\}.$$
It follows that the non-commutative rank of the operator $T$ is a property of
the matrix space $\cB$ and does not depend upon its presentation $T$. So it is natural to consider the problem
of determining the maximum $c$ such that $\cB$ has a $c$-shrunk subspace.

\paragraph{Rank decreasing operator:}
When the underlying field $\F$ is the field of complex numbers $\C$, given $B_1,\ldots,B_n$, consider the following positive operator $P$, 
$P: M(n, \C)\to M(n, \C)$, sending $A\to                                                       
\sum_{i\in[m]}B_iAB_i^{\dagger}$.  For $c \in \N$, the operator $P$ is said to be rank
$c$-decreasing if there exists a positive semidefinite matrix $A$ such that
$\rk(A)-\rk(P(A))=c$. Gurvits\cite{Gurvits} considered the problem of determining the maximum $c$ such that $P$ is rank $c$-decreasing.
It can be easily seen that $P$ is rank $c$-decreasing iff $\cB$ has a $c$-shrunk subspace - it was this formulation of the non-commutative rank
which Gurvits was interested, in his attempt to generalize the alternating minimization algorithm of  Linial, Samorodnistky and                                  
Wigderson~\cite{LSW00} for computing the permanent of a matrix. Gurvits proved that his algorithm runs in polynomial time 
when the commutative and non-commutative ranks of $\cB$ coincide.
\paragraph{The null cone for the left right action:}
Shrunk subspaces also appear naturally in a problem of classical invariant theory. 
Consider the action of $\SL(n, \F)\times \SL(n, \F)$ on $M(n, \F)^{\oplus m}$ with $(A,C)$ sending
a tuple $(B_1, \dots, B_m)$ to $(AB_1\trans{C},\ldots,AB_m\trans{C})$.\footnote{This action can
also be written as: $(A, C)$ sending $(B_1, \dots, B_m)$ to $(AB_1C^{-1}, \dots, AB_mC^{-1})$.}   Index the coordinates of the matrices by variables $(x_{i,j}^k)$, $1 \leq k \leq m$, $1 \leq i,j \leq n$.  Let  $R(n,m) \subseteq \F[x_{i,j}^{(k)}]$ be the $\F$-algebra
of polynomials in the variables $x_{i,j}^k$, invariant with respect to this action. In the literature this ring is also called the ring of matrix semi-invariants.
%This induces an action on                                                                                       
%polynomial functions on $M(n, \F)^{\oplus m}$; let $R(n, m)$ be the homogeneous                                 
%ring (algebra) consisting of those polynomials left unchanged under this action.                                
The \emph{nullcone} of $R(n, m)$ is locus of $m$-tuples of matrices where all homogeneous
positive-degree polynomials in $R(n, m)$ vanish. The null-cone is the set of points that need to be discarded when
one constructs the GIT quotient of the action of $SL(n,\C) \times SL(n,\C) $ on $m$-tuples of matrices.
This motivates the question of deciding whether an $m$-tuple $(B_1, \dots, B_m)$ is in the nullcone of
$R(n,m)$. Burgin and Draisma\cite{BD06} and, independently,  Adsul et al \cite{ANK07} showed that 
an $m$-tuple of matrices is in the null cone precisely when $\cB$ has a shrunk subspace. 

It is known that $R(n,m)$ is finitely generated and there is also a good description of the 
homogenous invariant polynomials, which follows from several independent works, 
including Derksen and Weyman \cite{DW00}, Schofield  and Van den Bergh 
\cite{SV01}, Domkos and Zubkov \cite{DZ01}, and Adsul et al \cite{ANK07}. 
Invariants exists only in degrees $nd$, as $d$ runs over all positive integers. To 
obtain invariants of degree $nd$ take matrices $A_1, \dots, A_m \in M(d, \F)$. 
Then $\det(A_1\otimes X_1+\dots+A_m\otimes X_m)$ is a matrix
semi-invariant, and every matrix semi-invariant of degree $nd$ is a linear combination of such
polynomials. Therefore $(B_1, \dots, B_m)$ is in the nullcone if and only if, for 
all $d\in\Z^+$ and all $(A_1, \dots, A_m)\in M(d, \F)^{\oplus m}$, $A_1\otimes B_1
+\dots+A_m\otimes B_m$ is singular. This motivates the following definition and leads us to the last formulation of the non-commutative rank.
\paragraph{Blow-ups:}
\begin{definition}
Given $\cB = \langle B_1, \ldots, B_m\rangle  \leq M(n, \F)$, the $d$th tensor
blow-up of $\cB$ is defined to be $\blowup{\cB}{d}:=M(d, \F)\otimes \cB\leq M(dn, \F)$, the 
linear span of matrices $A_1\otimes B_1
+\dots+A_m\otimes B_m$, with $A_i \in M(d, \F)$. 
\end{definition}
It is clear that $\rk(\blowup{\cB}{d})\geq d\cdot \rk(\cB)$. Furthermore, if $\cB$ 
has no shrunk subspace, then there is some $d$ for which $rk(\blowup{\cB}{d})=nd$; 
%\ynote{Added:} 
this follows from the descriptions of the nullcone and the 
invariants of the left 
right action.
Hence NCFullRank is equivalent to deciding whether $\rk(\blowup{\cB}{d})=nd$ for some $d$. This was also shown by  Hrube\v{s} and Wigderson \cite{HW15}. Hrube\v{s} and Wigderson's interest in knowing whether the 
non-commutative rank of a matrix family is full, was motivated by their study
of non-commutative arithmetic formulas \emph{with divisions}.  
In \cite{IQS1} we showed that when the field size $|\F|$ is large then $d$ divides $\rk(\blowup{\cB}{d})$. 
We refer to this as the regularity lemma, and defer the 
exact statement to a later point (\ref{lem_reg_blowup-window} in \ref{sec:reg}).
%(the exact 
%statement and
%the assumptions on $\F$ we made are deferred to a latter point). 

So, when $|\F|$ is large, we can define the {\em non-commutative rank} of $\cB$
to be the maximum over $d$ of $\frac{1}{d}$ times the maximum rank of a 
matrix from the blow-up $\rblowup{\cB}{d}$. 

From the last formulation above, an important question is to determine bounds on 
the blow-up parameter $d$ (as a function of $n$) which achieves the desired 
maximum. We define $\sigma(R(n, m))$ to be the smallest $d\in \N$, 
such that those 
non-constant homogeneous invariants of degree $\leq d$ define the nullcone of 
$R(n, m)$.
%\gabor{Copied the previous sentence from Youming's footnote.} 
%\ynote{Added: We define $\sigma(R(n, m))$ to be the smallest $d\in \N$, 
%such that those 
%non-constant homogeneous invariants of degree $\leq d$ define the nullcone of 
%$R(n, m)$.} 
From the work of 
Derksen \cite{derksen_bound} it follows that $\sigma(R(n, m)) \leq  O(n^4\cdot 
4^{n^2})$, over 
algebraically 
closed fields of characteristic 
zero.\footnote{Derksen's result applies to a wide class of invariant rings.} In 
\cite{IQS1} we showed that $\sigma(R(n, m))\leq 2^{O(n\log n)}$
%\ynote{Was: $d = O(n^4\cdot 2^{O(n\log n)})$}
over large fields 
of arbitrary characteristic. We also gave an algorithm to compute $\ncrk(T)$ and output a witnessing shrunk subspace with running time 
%$O(n^7 \cdot 2^{O(n\log n)})$ 
$2^{O(n\log n)}$ over large fields.

We describe this algorithm in the next section. After that we describe further progress on the non-commutative rank from the works of Garg et al~\cite{GGOW}
and Derksen and Makam~\cite{DM2}. We then state the main theorem of the paper. 

\subsection{Outline of the algorithm in \cite{IQS1}}\label{subsec:old_algo}

The algorithm in \cite{IQS1} can be viewed as an analogue of the augmenting path 
algorithm for the bipartite maximum matching problem. However, due to the 
failure of the analogue of Hall's marriage theorem in the matrix space setting, 
there are a couple of new and sophisticated 
components. 

%\kvnote{removed}
%\eat{
Let us briefly review some features of the augmenting path algorithm. %} 
Given a matching $T$ 
for the input bipartite graph $G=(L\cup R, E)$, the algorithm tries to find an 
augmenting path for $T$. If an augmenting path is found, $T$ is replaced by a 
larger matching $T'$. If no augmenting paths can be found, the algorithm can 
output a shrunk subset as the certificate of the maximality of $T$. 

We hope to implement the above idea for the non-commutative rank problem. Given 
a matrix $A\in \cB=\linspan(B_1, \dots, B_m)\leq M(n, \F)$, we would like to 
either find an ``augmenting 
path'' for it and increase its rank, or output a $c$-shrunk subspace where 
$c=\cork(A)$. 

%A linear algebraic analogue of 
%augmenting paths was first developed by Fortin and Reutenauer \cite{FR04}, though 
%they seemed not to realize the connection with augmenting paths. This 
%was rediscovered by Ivanyos et al
%\cite{conf_version}. 

A linear algebraic analogue of augmenting paths was developed in 
\cite{conf_version}.
Given a subspace $U\leq \F^n$, let $A^{-1}(U)$ be the 
preimage of $U$ under $A$, namely the subspace $\{v\in \F^n : A(v)\in U\}$. We 
also define $\cB(U):=\linspan(\cup_{i\in [m]} B_i(U))$. Given $A\in \cB\leq M(n, 
\F)$, we apply $\cB$ and $A^{-1}$ iteratively to $V_0=\ker(A)$, to get 
$W_1=\cB(V_0)$, $V_1=A^{-1}(W_1)$, $W_2=\cB(V_1)$, \dots, $V_i=A^{-1}(W_i)$, 
$W_{i+1}=\cB(V_i)$, \dots. It can be shown that for some $\ell\in[n]$, 
$W_1<W_2<\dots<W_\ell=W_{\ell+1}=\dots$. This sequence of subspaces is called \emph{the second Wong 
sequence} of $(A, \cB)$.
%\gabor{Changed footnote on naming.}\ynote{Merged two footnotes into one.}
\footnote{The first 
Wong sequence is the dual of the 
second one. The sequences are named after Wong who 
defined them in~\cite{Wong} for the special case when $\cB$ is of dimension $1$. Over $\Q$ the 
straightforward 
implementation of the second Wong sequence may lead to a bit size explosion. To 
avoid that some tricks are needed. See \cite{conf_version} for more details.} 
$W_\ell$ is called 
the \emph{limit subspace} of this sequence.  We state as a fact the following important lemma from \cite{conf_version}.

\begin{fact}[{\cite[Lemmas 9 and 10]{conf_version}}]
\label{fact:Wong}
Let $A \in \cB \leq M(n,{\F})$, and let $W^{*}$ be the limit of the second Wong
sequence of
$(A, \cB)$. Then there exists a $\cork(A)$-shrunk subspace of $\cB$ if and only
if $W^{*} \leq  \im(A)$.
\footnote{At the time of
writing the first version of \cite{conf_version}, the authors
were unaware of \cite{FR04} where this  had already appeared.}
%\gabor{Put back footnote.}
If this is the case then $A^{-1}(W^{*})$ is a
$\cork(A)$-shrunk subspace of
$\cB$. In the algebraic RAM model, as well as over ${\mathbb Q}$, we can detect
whether $W^*\subseteq \im(A)$, and in that case we can compute a shrunk subspace
in deterministic polynomial time.
\end{fact}
Therefore 
when $W_\ell\leq \im(A)$, we can conclude that the non-commutative rank is $\rk(A)$. 
On the other hand, when $W_\ell\not\leq \im(A)$, following the bipartite maximum 
matching 
algorithm it seems natural to try to obtain $A'\in \cB$ with $\rk(A')>\rk(A)$. 
However this is not always possible, as it can be the case that 
$\rk(A)=\crk(\cB)$ and $\crk(\cB)<\ncrk(\cB)$. 
But for a matrix space $\cB$ of dimension $2$, $\rk(\cB)=\nrk(\cB)$ for large
enough $\F$; this follows from the Kronecker-Weierstrass theory
of matrix pencils --
alternate proofs may be found in \cite{EH88,PrimitiveI}. 

The key observation in \cite{conf_version} was
that, in certain special cases, when $W_\ell\not \leq \im(A)$ the second
Wong sequence could be used
to find an ``augmenting'' matrix $B$ from $\cB$ such that
$\rk(\mu A+\lambda B)>\rk(A)$ for
some scalars $\lambda$ and $\mu$. This included the case of two-dimensional matrix spaces. 
The authors showed
\begin{fact}[{\cite[Fact 11]{conf_version}}]
\label{fact:dimtwo}
Assume that $|\F|>n$, and let $\cB=\langle A, B\rangle \leq M(n, \F)$. Then
$\rk(A)=\rk(\cB)$ if and only if for any $i\in[n]$, $(\cB                                                 
A^{-1})^i(\zvec)\leq                                                                                      
\im(A)$.
\end{fact}

The key idea in \cite{IQS1}
is to reduce the general problem to the rank two situation. The idea is to find $A'\in \rblowup{\cB}{d}$ of rank 
$\geq (r+1)d$ with some not too large $d$ (so that the scaled-down rank 
$\rk(A')/d$ is larger than $r$), and iterating this procedure.  We give the key steps of that algorithm.

\paragraph{A: Incrementing the scaled-down rank.} This is achieved in two steps.
\begin{itemize}
\item[1] {\sf Incrementing rank:} 
\label{step1}
The first step 
is to obtain a matrix $\widehat{A}\in \rblowup{\cB}{d}$ of rank $\geq 
rd+1$ where $d=r+1$.  To see how this step works, notice first that by multiplying 
$A$ and $\cB$
with an appropriate matrix, one can arrange $A$ to be idempotent. In that
case, as long as $W_1,\ldots,W_{j-1}$ remain inside $\im(A)$,
we have $W_j=\cB^j\ker(A)$. Let $l$ be the smallest index $j$ with
$W_j\not\leq \im(A)$. Then $l \leq r+1$.
Then there exist matrices 
$B_1, B_2,\ldots,B_l$ such that $B_l\cdots B_1\ker(A)\not\leq \im(A)$.
It would be nice if one could find a {\em single} matrix $B\in \cB$
such that $B^l\ker(A)\not\leq  \im(A)$: indeed if this happens then for some 
$\lambda$ and $\mu$ from a subset of the base field of size at least $r+1$ one 
would 
have for $\widehat{A} = \mu A+\lambda B$, $\rk(\widehat{A}) >\rk(A)$. 
This follows from Fact~\ref{fact:dimtwo}.
%, a result equivalent to a theorem of Atkinson and Llyod \cite{PrimitiveI}. 

The main ingredient of
the algorithm in \cite{conf_version} was a method to find such
a $B\in \cB$ in certain special cases. The idea in \cite{IQS1} is
that, if we relax ourselves to work with $\rblowup{\cB}{d}$, then this can be 
achieved for every matrix space $\cB$. 

\item[2] {\sf Rounding up the rank:} 
\label{step2} For the second step, starting with $\widehat{A}$, we wish to get the desired 
$A'\in\rblowup{\cB}{d}$ of rank $\geq (r+1)d$. This is accomplished  in ~\cite{IQS1} by the 
regularity lemma. An efficient, constructive version of this lemma is required in the algorithm. And to accomplish this we need an efficient construction of central division algebras of degree $d^2$ over $\F$ with an explicit matrix representation of such a division algebra.  In ~\cite{IQS1} we were able to 
construct explicit division algebras when the characteristic of $\F$ and $d$ are coprime. 

We reproduce the constructive regularity lemma from ~\cite{IQS1} below. 

%This lemma was proved under some assumptions on the characteristic of the underlying field and the field size. In this paper we
%state an improved regularity lemma  where the only assumption is on the field size - see Lemma~\ref{lem_reg_blowup} in Section~\ref{subsec:new_algo}). 

{\bf Lemma 5.7} in \cite{IQS1} (Regularity of blow-ups, constructive).
% removed a label from here
For $\cB\leq M(n, \F)$ and $\cA=\rblowup{\cB}{d}$, assume that
$\mathrm{char}(\F)=0$ or
$\mathrm{char}(\F)\nmid d$, and
$|\F| > (nd)^{\Omega(1)}$. Then,
given a matrix $A\in                                                                               
\cA$ with $\rk A > rd$, there exists a deterministic algorithm that returns
$\widetilde{A}\in \cA$ of rank $\geq (r+1)d$.
This algorithm uses $\poly(nd)$ arithmetic
operations and over $\Q$, all intermediate numbers have bit lengths polynomial in
the input size.
 
 This $A'\in \rblowup{\cB}{d}$ of rank $\geq (r+1)d$ where $d=r+1$ certifies 
that $\ncrk(\cB)\geq r+1$. (From the viewpoint of shrunk subspaces, it is easy to 
see that $\ncrk(\cB)\leq r$ then $\crk(\rblowup{\cB}{d})\leq 
rd$ for any $d$; see e.g. \cite[Proposition 5.2]{IQS1}.) So after these two 
steps we obtain $A'$ of rank $r'd$ where $r'>r$. 

\end{itemize}

\paragraph{B: Iterating over.}
\label{step3} In the next phase, we need 
to use $A'$ and $\rblowup{\cB}{d}$ to restart the above procedure, hoping either 
to 
find a $\cork(A')$-shrunk subspace, or to obtain some $A''$ in 
$\rblowup{\cB}{dd'}$ 
of rank $r''dd'$ where $r''>r'$. We then apply the second Wong sequence to work 
with the blow-up space $\rblowup{\cB}{d}$ and $A'$.\footnote{When the second Wong 
sequence is applied to such blow-up spaces then it has some nice properties; cf. 
the proof for Theorem 5.10 in \cite{IQS1}.} If $\cork(A')$-shrunk subspace $U'$
is found for $\rblowup{\cB}{d}$, then this naturally induces a 
$\cork(A')/d$-shrunk 
subspace $U$ for $\cB$ \cite[Proposition 5.2]{IQS1}. In this case we conclude that 
the non-commutative rank is $r'$, and $A'$ and $U$ together serve as witnesses for 
this fact. If the limit subspace goes out of $\im(A')$ we need to go to an even 
larger blow-up space $\rblowup{(\rblowup{\cB}{d})}{d'}\cong \rblowup{\cB}{dd'}$ 
where 
$d'=r'+1$, to find a matrix $A''\in \rblowup{\cB}{dd'}$ of rank $r''dd'$ for some 
$r''>r'$. 

We reproduce the following theorem from \cite{IQS1} which summarizes the above 
discussion.
\vskip .5em
\noindent{\bf Theorem 5.10} in \cite{IQS1}.
%\label{thm-blup-incr-old}
Let $\cB\leq M(n, \F)$ and $\cA=\rblowup{\cB}{d}$. Assume that
we are given a matrix $A\in \cA$ with $\rk(A)=rd$. Let $d'$ be an integer $>r$.
Suppose that $|\F|$ is $(ndd')^{\Omega(1)}$, and if $\fdchar(\F)=p>0$ then assume
$p\nmid                                                                                            
dd'$. There
exists a deterministic algorithm that returns either an $(n-r)d$-shrunk
subspace for $\cA$ (equivalently, an $(n-r)$-shrunk subspace for $\cB$), or a
matrix $A^*\in \cA\otimes M(d', \F)$ of rank at least $(r+1)dd'$. This algorithm
uses $\poly(ndd')$ arithmetic operations and, over $\Q$, all intermediate numbers
have bit lengths polynomial in the input size.
\vskip .5em
The main point is that to carry out the 
augmenting path idea for the bipartite 
maximum matching problem in the non-commutative rank setting, the right approach 
is to play with shrunk subspaces on the one hand, and matrices in the blow-up spaces 
on the other. 

The alert reader may now notice that the above strategy leads to an 
exponential-time algorithm. Recall that we start with $A\in \cB$ of rank $r$. If 
$\ncrk(\cB)=n$, then we may end up finding $A^*\in \rblowup{\cB}{d^*}$ of rank 
$nd^*$ where $d^*$ can be as large as $n!/r!$. This is because, increasing the 
scaled-down rank from $r'$ to $r'+1$ would lead to a 
multiplicative factor of $r'+1$ in the size of the blow-up space. This is 
why the algorithm in \cite{IQS1} runs in time $\poly(n!)$.  We reproduce that result
below.
\vskip .5em
\noindent{\bf Theorem 5.11} in \cite{IQS1}.
%\label{thm-blup-main}
Suppose we are given $\cB:=\langle B_1, \dots, B_m\rangle\leq M(n, \F)$, and
$A\in\cB$ with $\rk(A)=s<n$. Let $d=(n+1)!/(s+1)!$, and assume that
$|\F|=\Omega(nd)$. Then there exists a deterministic algorithm, that
computes a matrix $B\in \cB\otimes M(d', \F)$ of rank $rd'$ for some $d'\leq d$
and, if $r<n$, an $(n-r)$-shrunk subspace for $\cB$. The
algorithm uses $\poly(n, d)$ arithmetic operations, and when working over $\Q$,
has bit complexity polynomial in $n$, $d$ and the input size.

\subsection{Progress on non-commutative rank since 2015.}
\label{sec:progress}

%\ynote{Added:} 
Recall that an important question was to upper bound $\sigma(R(n, 
m))$, and exponential bounds were established in 
\cite{derksen_bound} and \cite{IQS1}. These turned out to be 
sufficient for~\cite{GGOW} to 
compute the non-commutative rank in deterministic polynomial time,
over fields of characteristic zero, %\ynote{Added:} 
by a more refined analysis of 
Gurvits' algorithm in \cite{Gurvits}.
%\ynote{Deleted: The algorithm of ~\cite{GGOW} was based on the earlier mentioned 
%result of
%Gurvits \cite{Gurvits}.}
After \cite{GGOW}, the following problems were still
open: 
\begin{enumerate}
\item[(1)] 
a 
polynomial-time algorithm for the problem over finite 
fields, and 
\item[(2)] a search version of the problem, that is, explicitly exhibiting a 
matrix of rank $rd$ in the $d$-th blow-up and a proof that the non-commutative 
rank is at most $r$, even over fields of 
characteristic 0. 
\end{enumerate}

Recently, Derksen and Makam\cite{DM2}
proved that it suffices to take the maximum over  $d$ between $1$ and $n-1$, for sufficiently large fields, by discovering a 
concavity property of blow-ups, and using the regularity lemma of blow-ups from 
\cite{IQS1}. In the first version of this note, by 
showing that the concavity property can be made constructive, and building on the 
techniques from \cite{IQS1}, we obtained a deterministic polynomial-time algorithm 
for the non-commutative rank problem, which is constructive and works over large 
enough fields regardless of the characteristic. This answers the two open 
problems just mentioned. 

After the first version of this note appeared on the arXiv, we discovered that a very simple 
observation already gives us the result, without having to use the results from 
Derksen and Makam. This argument also gives a different proof that the nullcone of 
the matrix semi-invariants 
%\ynote{Deleted:(for this concept, and a proof, see \cite{IQS1,DM2})}
is 
generated by polynomials in $R(n,m)$ of degree less than or equal to $O(n^2)$.  We should point out 
that recently Derksen and Makam~\cite{DM3} also gave a %\ynote{Added:} 
second proof 
of 
the regularity lemma. However their proof is not %\ynote{Added:} 
known to be 
constructive.

We now state our main result and the contributions of this paper.

\begin{theorem}
\label{thm:main}
Let $\cB\leq M(n,\F)$ be a matrix space given by a linear basis, and suppose 
$|\F|=n^{\Omega(1)}$. 
%the field of size $n^{\Omega(1)}$ given by a basis over $\F$.
Suppose that $\cB$ has (a priori unknown) non-commutative rank $r$. 
Then there is a deterministic algorithm
using $n^{O(1)}$ arithmetic operations over $\F$ that constructs
a matrix of rank $rd$ in a blow-up $\rblowup{\cB}{d}$ 
for some $d\leq r+1$ 
as well as an $(n-r)$-shrunk subspace 
of $\F^n$ for $\cB$. 
When $\F=\Q$, the final data as well  as all the intermediate
data have size polynomial in the size of the input data
and hence the algorithm runs in polynomial time.
\end{theorem}

Compared with the algorithm in \cite{GGOW}, our algorithm has 
the 
advantages of (1) working with arbitrary large enough fields, and (2) outputting a 
shrunk subspace and a matrix in a blow-up space certifying 
that the non-commutative rank is $r$. Note that the second feature is new even 
over $\Q$.  We also show that the small 
finite fields case can be handled as well.

\begin{remark}
\begin{enumerate}
\item[(a)]
If the constructivized version of Derksen and Makam \cite{DM2} is used, 
then in the above theorem we can 
improve the parameter slightly to $d\leq r-1$ instead of $d\leq r+1$.
\item[(b)]
Polynomial running time of the algorithm can also be proved for
a wide range ``concrete'' base fields $\F$. These include 
sufficiently large finite fields, and also number fields
and transcendental extensions of constant degree
over finite fields and over number fields.
\item[(c)] In particular, the non-commutative rank can be computed in 
deterministic polynomial time in positive characteristic 
as well, assuming that the ground field is sufficiently 
large.
%This then settles the aforementioned two open problems from \cite{GGOW}.
\end{enumerate}
\end{remark}

Our result also settles a question of Gurvits~\cite{Gurvits},
asking if it is possible to decide efficiently, over fields of positive 
characteristic,
whether or not there exists a non-singular matrix in a matrix space 
having the Edmonds-Rado property. Recall that a matrix space has the 
{\em Edmonds-Rado property} if it satisfies the promise that it either contains a 
non-singular matrix, or it
 shrinks some subspace. Since the algorithm in ~\ref{thm:main} efficiently 
 tells whether the given 
matrix space has 
a shrunk subspace (e.g. the non-commutative rank is not full), it settles Gurvits' 
question, when the field size is as stated in the hypothesis.

\paragraph{Over small finite fields.}
%\ynote{Added:} 
From the above, we have seen a polynomial upper bound on $\sigma(R(n, m))$,  and 
settled the non-commutative rank problem as well as SDIT for the Edmonds-Rado 
class, 
provided that the underlying field is large enough.
However we can say more, even when the base field is a ``too small'' finite field. 

\begin{corollary}
\label{cor:small}
Let $\F$ be a finite field of size $s<n^{O(1)}$. 
\begin{enumerate}
\item Let $R(n, m)$ be the ring of matrix semi-invariants over $\F$. Then 
$\sigma(R(n, m))\leq O((n^2-n)\log_s n).$
\item Let $\cB\leq M(n,\F)$ be a matrix space given by a linear basis with a 
priori unknown non-commutative rank $r$. 
There is a deterministic polynomial-time algorithm
that constructs
a matrix of rank $rd$ in a blow-up $\rblowup{\cB}{d}$ 
for some $d\leq O(r\log_s n)$, as well as an $(n-r)$-shrunk subspace 
of $\F^n$ for $\cB$. 
\item Let $\cB\leq M(n,\F)$ be a matrix space given by a linear basis satisfying 
the Edmonds-Rado property. Then there exists a deterministic polynomial-time 
algorithm that can decide whether $\cB$ has a non-singular matrix, or a 
shrunk-subspace. 
\end{enumerate}
\end{corollary}

\paragraph{Techniques.} 
As described in the the iterating  over step in Section\ref{step3}, the algorithm in \cite{IQS1} takes exponential time
because we increase the blow-up size in an iterative way, and in each iteration 
the blow-up size is increased multiplicatively by the ``scaled'' rank. 
%\kvnote{added this} 
The key new 
insight is that we can keep the blow-up size small: when the scaled rank is $r$, 
then 
the blow-up size can be brought back to $O(r)$. As mentioned, we offer two methods 
to realize this reduction idea: a simpler method from us, and a method based 
on the technique of Derksen and Makam~\cite{DM2}. 

We also provide a technical improvement to the constructive regularity lemma used in the rounding up the rank step of the algorithm
described in \ref{step2}. 
Recall that we use it in the algorithm in the following situation: given $A\in \cB\otimes M(d, \F)$ of rank 
$(r-1)d+k$ where $1<k<d$, we want to construct $A'\in \cB\otimes M(d, \F)$ of rank 
$\geq rd$ efficiently.  This was achieved under the 
condition that, if $\fdchar(\F)=p>0$, then $p$ and $d$ are coprime. In this note, we remove this coprime 
condition. 
%This is achieved by presenting an efficient construction of cyclic field 
%extensions of arbitrary degrees. 

\paragraph{Organization.} In Section \ref{sec:algo-issues} we first discuss algorithmic issues that arise when working over finite extensions of fields and how they are solved. Since all this appears with detailed proofs in our previous paper we only provide pointers to these issues and refer to \cite{IQS1} for details.
In  Section \ref{sec:cyclic} we give an efficient construction of cyclic field extensions of arbitrary degrees.  In  Section~\ref{sec:reg} we use this to prove the full regularity lemma. In Section \ref{subsec:new_algo} we prove the main Theorem~\ref{thm:main} using our blow-up 
reduction method. In Section~\ref{sec:small}  we give the proof for Corollary\~ref{cor:small}. Finally 
in  Section~\ref{sec:dm} we show that the Derksen--Makam technique can be 
constructivized to provide another blow-up reduction method. 

\section{Preparations on certain algorithmic issues}
\label{sec:algo-issues}
In this section we highlight algorithmic issues which need to be addressed to ensure that our algorithms run in polynomial time. All these issues have been addressed in our earlier paper. So we only indicate briefly where these issues arise and what needs to be done. For details and proofs the really interested reader should refer to \cite{IQS1}.

\paragraph{From the extension field to the original field.}
Assume that for some extension field $\K$ of $\F$ we are given a
matrix $A'\in \cB \otimes_\F \K\leq M(n, \K)$ of rank $r$. Then, if $|\F|> r$,
using the method of \cite[Lemma 2.2]{GIR}, we can efficiently find a matrix $A\in                                      
\cB$
of rank at least $r$. This procedure is also useful to keep sizes of
the occurring field elements small. This is how it gets used in Lemma \ref{lem:reg_technical} and in Theorem~\ref{thm:main}.
We give details for this procedure alone.

Let $S\subseteq \F$ with $|S|=r+1$ and
let $B_1,\ldots,B_\ell$ be an $\F$-basis for $\cB$. Then
$A'=a_1'B_1+\ldots+a_\ell'B_\ell$, where $a_i'\in\K$. As $A'$ is of rank $r$,
there
exists an $r\times r$ sub-matrix of $A$ with nonzero determinant. Assume that
$a_1'\not \in S$. Then
we consider the determinant of the corresponding sub-matrix of the polynomial
matrix $xB_1+a_2'B_2+\ldots a_{\ell}'B_\ell$. This determinant is
a nonzero polynomial of degree at most $r$ in $x$. Therefore there exists
an element $a_1\in S$ such that $a_1B_1+a_2'B_2+\ldots a_{\ell}'B_\ell$ has rank
at least $r$. Continuing with $a_2',\ldots,a_{\ell}'$, we can ensure
that all the $a_i$'s are from $S$. Since the $B_i$'s span $\cB$, the resulting
matrix of rank at least $r$ is
in $\cB$. We record this as a fact.

\begin{lemma}[{Data reduction, \cite[Lemma 2.2]{GIR}}]
\label{lem:red_data}
Let $\cB\leq M(k\times \ell,\F)$ be given by a basis
$B_1,\ldots,B_m$, and let $\K$ be an extension field of $\F$. Let $S$ be a subset
of $\F$
of size at least $r+1$. Suppose that we are
given a
matrix $A'=\sum a_i'B_i \in \cB\otimes_{\F}\K$ of rank at least $r$.
Then
we can find $A=\sum a_i B_i\in \cB$ of rank also at least $r$
with $a_i\in S$. The algorithm uses $\poly(k, \ell, r)$
rank computations for matrices of the form $\sum a_i''B_i$
where $a_i''\in \{a_1',\ldots,a_m'\}\cup S$.
\end{lemma}

\paragraph{Dealing with the need for a primitive root of unity.}
Lemma \ref{lem:cyclic_gen} assumes the field $\F'$ contains a
known primitive $d$th root of unity
$\zeta$. In actual applications, we start with a field $\F$ without a primitive
$d$th root of unity in it, and attach one symbolically,
%\ynote{Modify slightly in the next few 
%words.}                                                                  
which we still denote by
$\zeta$. However, this
%(denoted as 
%$\zeta$)                                                                          
%                  
may cause some problem. Namely, constructing
$\F'=\F[\zeta]$ would require factoring
the polynomial $x^{d}-1$ over $\F$, a task which cannot be
accomplished using basic arithmetic operations. To see that
this is indeed an issue notice that
 a black-box field may contain certain ``hidden'' parts
of cyclotomic fields.
Of course,
over certain concrete fields, such as the rationals,
number fields or finite
fields of small characteristics, this can be done in polynomial time.
However, even over finite fields of large characteristic no
deterministic polynomial time solution to this task is known
at present.
%with a mere arithmetic circuit.      

To get around this issue, one can
perform the required computations over an appropriate factor algebra $R$ of
the algebra $C=\F[x]/(x^{d}-1)$ in place $\F'$ as if $R$ were a field.
To be specific, as $d$ is not divisible by the characteristic,
we know that $C$ is semisimple -- actually it is isomorphic to
a direct sum of ideals, each of which is isomorphic
to the splitting field $\F[\sqrt[e]{1}]$ of the polynomial $x^{e}-1$ for
some divisor $e$ of $d$, and the projection of $x$ to such an ideal
is a primitive $e$th root of unity. It follows that if we compute the ideal $J$ 
generated by
annihilators of $x^e -1$, for all $e$ a proper divisor of $d$, then $R=C/J$ is 
isomorphic to the direct
sum of copies of the splitting field $\F'$ of $x^{d}-1$, and the projection
of $x$ to each component is a primitive $d$th root of unity. And this
property is inherited by any proper factor of $R$. A
computation using $R$ instead of $\F'$ may fail
only at a point where we attempt to invert an non-invertible
element of $R$. However, such an element must be a zero divisor.
%\kvnote{moved the next sentence to the remark. Youming: Add back since 
%that                                     
%sentence points out how to deal with zero divisors in general, and seems 
%needed                                 
%here.}.                                                                           
%                              
When this situation occurs, we replace $R$
with the factor of $R$ by its ideal generated by the zero divisor and
restart the computation. Such a restart can clearly happen at most $d-2$ times.

%\paragraph{Computing rank of matrices over $\F'=\F[\zeta]$.}
Now consider the task of computing the rank of a matrix in $M(n, \F')$. As 
described above we work instead with coefficients in $R$. 
Note that we cannot
talk about the ``rank'' of matrices in
$M(n, R)$ which is not well-defined. But since $R$ is a direct sum of $\F'$, the
decomposition of $R$ induces a decomposition of $M(N, R)$ into a direct sum of
copies of $M(N, \F')$.  
We call the images of the projections of a matrix
$B\in M(N, R)$ to the direct summands 
the {\em components} of $B$. The following lemma from \cite{IQS1} describes how to 
compute the 
maximum rank over the components. 
\begin{lemma}[{\cite[Lemma 4.6]{IQS1}}]
\label{lem:correct_rank}
Let $R$ and $\F'$ be as above, and suppose we are given a matrix $B\in M(N, R)$. 
Then there exists a deterministic polynomial-time algorithm that computes the 
maximum rank over the components of $B$. 
\end{lemma}

We remark that the issue with the need of roots of unity
and working over rings instead of fields occurs only when we apply
the algorithm for the constructive regularity lemma. It has no
influence of the other parts of the algorithm, as 
after having constructed a matrix over the ring 
$R$ having sufficiently large ``rank'', we can apply
Lemma~\ref{lem:red_data} to obtain a matrix over the
base field $\F$ of the same or larger rank, provided that $\F$
is large enough. (Cyclotomic
extension fields of finite fields  
can be constructed deterministically in time
polynomial in the field size, so over small fields
such issues do not occur at all.) %\gabor{Added the above paragraph}

%In  \cite[Lemma 4.6]{IQS1}  we show how to compute the maximum rank over 
%such a direct sum by combining the algorithms of Kaltofen \cite{Kal92} and 
%Mulmuley\cite{Mul87}.

\paragraph{Computing the rank of matrices over a rational function field in few 
variables.}
In  Lemma \ref{lem:reg_technical} we will need to compute the rank of matrices over a rational function field of $\F'$ in two variables.
The following proposition from ~\cite{IQS1} describes how when the field size $\F'$ is large we can find a matrix 
over the base field with the same rank as the matrix we start with. 

\begin{proposition}[{\cite[Lemma 4.8]{IQS1}}]
\label{prop:rank}
Let $\F'$ be a field and $\K=\F'(X_1,\dots, X_k)$ be a pure transcendental
extension of $\F'$. Let $A$ be an $N \times N$ matrix with entries as quotients of
polynomials from $\F'[X_1, X_2, \dots, X_k]$, where the polynomials are explicitly
given as sums of monomials. Assume that the degrees of the
polynomials appearing in $A$ are upper bounded by $D$. If
$|\F'|=(ND)^{\Omega(k)}$, then we can find in
time $(ND)^{O(k)}$ a matrix $B\in M(N, \F')$ with $\rk(B)=\rk(A)$.
\end{proposition}

\section{Efficient construction of cyclic field extensions of arbitrary degrees}
\label{sec:cyclic}
A cyclic extension of a field $\K$
is a finite Galois extension of $\K$ having a cyclic Galois group.
By constructing a cyclic extension $\ELL$ we mean constructing
the extension as an algebra over $\K$, e.g., by giving an
array of {\em structure constants} with respect to a $\K$-basis for $\ELL$
defining the multiplication on $\ELL$ as well as specifying a generator 
of the Galois group, e.g, by its matrix with respect to a $\K$-basis.

\begin{lemma}
\label{lem:cyclic_p}
Given a prime $p$ and an integer $s\geq 1$,
one can construct in time $\poly(p^s)$ 
a cyclic extension $K_s$ 
of $\F_p(Z)$ of degree $p^s$ such that $\F_p$ is 
algebraically closed in $K_s$. The field $K_s$
will be given in terms of structure constants
with respect to a basis over $\F_p(Z)$, and the
generator $\sigma$ for the Galois group will be given
by its matrix in terms of the same basis. The structure
constants as well as the entries of the matrix for $\sigma$
will be polynomials in $\F_p[Z]$ of degree $\poly(p^s)$.
\end{lemma}

\begin{proof}
First we briefly recall the general
construction given in Section~6.4 of \cite{Ramanathan}.
This, starting from a field $K_0$ of characteristic $p$, recursively 
builds a tower  $K_0<K_1<\ldots<K_{s}$ of fields such that
$K_j$ is a cyclic extension of $K_0$ of degree $p^j$. Assume that $K_{s}$ 
together
with a $K_0$-automorphism $\sigma_{s}$ of order $p^{s}$
has already been constructed. (Initially let $\sigma_0$ be the identity map on
$K_0$.) Then 
for any element $\beta_{s}\in K_{s}$ with
$\Tr_{K_{s}:K_0}(\beta_s)=1$ and for any $\alpha_{s}\in K_{s}$
such that $\alpha_{s}^{\sigma_{s}}-\alpha_{s}=\beta_{s}^p-\beta_{s}$
the polynomial $X^p-X-\alpha_{s}$ is irreducible in $K_{s}[X]$.
(Existence of $\alpha_{s}$ with
the required property follows from the additive Hilbert~90.)
Put $K_{s+1}=K_{s}[X]/(X^p-X-\alpha_{s})$ and let $\omega_{s+1}\in K_{s+1}$ 
be the image of of $X$ under the projection $K_{s}[X]\rightarrow K_{s+1}$.
Then $\sigma_{s}$ extends to a $K_0$-automorphism $\sigma_{s+1}$
of degree $p^{s+1}$ of $K_{s+1}$ such that 
$\omega_{s+1}^{\sigma_{s+1}}=\omega_{s+1}+\beta_s$. 
This gives a cyclic extension of degree $p^{s+1}$. 

Now we specify some details of a polynomial time
construction for $K_0=\F_p(Z)$ following the method outlined above. 
In the first step we take $\beta_0=1$, and, in order to guarantee 
that the only elements in $K_1$ which are algebraic over $\F_p$ is $F_p$ (we also 
use the phrase $F_p$ is algebraically closed in $K_1$ when this property holds), 
we take $\alpha_0=Z$. Then $K_1$ is a pure transcendental
extension of $\F_p$. As $K_s/K_0$ is a cyclic extension of oder $p^s$,
it has a unique subfield which is an order $p$ extension of $K_0$. 
This must be $K_1$. Then $\F_p$ has no proper finite extension in 
$K_s$ as otherwise $K_0$ would also have another degree $p$ extension.

We consider the following $K_0$-basis for $K_s$: 
$$\Gamma_s=\left\{\prod_{j=1}^s\omega_j^k,\;\;\;(k=0,\ldots,p-1)\right\},$$
where $\omega_j$ is a root of $X^p-X-\alpha_{j-1}$ in $K_j$.
We claim that $\Tr_{K_j:K_{j-1}}(\omega_j^{p-1})=-1$. Indeed,
in the $K_{j-1}$-basis $\omega_j^0,\ldots,\omega_j^{p-1}$ for $K_j$,
in the matrix of multiplication by $\omega_j^{p-1}$ the diagonal entries 
consist of $p-1$ ones and one zero. Therefore
$\Tr_{K_j:K_{j-1}}(\omega_j^{p-1}
\gamma)=-\gamma$ for every $\gamma\in K_{j-1}$, whence
$\Tr_{K_j:K_0}(\omega_j^{p-1}\gamma)=-\Tr_{K_{j-1}:K_0}(\gamma)$.
Now by induction we obtain
$\Tr _{K_j:K_0}
\prod_{i=1}^j\omega_i^{p-1}=(-1)^j$.
Therefore in each step (when $j>0$) we can choose 
$\beta_j=(-1)^j\prod_{i=1}^j\omega_i^{p-1}$
and $\alpha_j$ thereafter, following the construction in the standard
proof of the additive Hilbert 90. Specifically,
we set
\begin{equation}\label{eq:alpha}
\alpha_j=
(-1)^{j+1}\sum_{k=1}^{p^j-1}\beta_j^{\sigma_j^k}\left(
\sum_{\ell=0}^{k-1}(\beta_j^p-\beta_j)^{\sigma_j^\ell}
\right).
\end{equation}
Then $\alpha_j^{\sigma_j}-\alpha_j=\beta_j^p-\beta_j$. 
Notice that $\alpha_j$ 
is a sum of terms with each of which, up to a sign, 
is a product of at most $p+1$ conjugates $\beta_j^{{\sigma_j^\ell}}$
(with various $\ell$s)
of $\beta_j$ ($\ell\leq p^{j}$) 

Assume by induction that the structure constants of $K_{j}$ with respect
to the basis $\Gamma_{j}$ are polynomials from $\F_p[Z]$ of degree
at most $\Delta_j$ and the same holds for the entries
of the matrix of $\sigma_j^{\ell}$ for every $1\leq \ell <p^j$ (written
in the same basis). For $j=1$ this holds with $\Delta_1=1$.
(To see this, observe that for $0\leq k,\ell<p$, 
the product
$\omega_1^{k}\omega_1^{\ell}$ is the basis element of
$\omega_1^{k+\ell}$ if $k+\ell<p$, while otherwise it 
equals the sum $\omega_1^{k+\ell-p+1}
+Z\omega_1^{k+\ell-p}.)$
Then, if we express $\alpha_{j}$ in terms of
the basis $\Gamma_j$ using Eq.~\ref{eq:alpha}, we obtain that its coordinates
are polynomials of degree at most $(2p+1)\Delta_j$. This is because
$(-1)^j\beta_j\in \Gamma_j$, whence $\beta_j^{\sigma^\ell}$ has
coordinates of polynomials of degree bounded by $\Delta_j$. In Eq.~\ref{eq:alpha}, 
we have the products of at most $p+1$ such elements, so the result will
have polynomial coordinates of degree at most $(2p+1)\Delta_j$. 

Now consider the product of two elements
$\omega_{j+1}^k\gamma_1$ and $\omega_{j+1}^\ell\gamma_2$
of $\Gamma_{j+1}$. Here $k,\ell<p$ and $\gamma_1,\gamma_2\in \Gamma_j$.
The coordinates of the product $\gamma_1\gamma_2$ with respect 
to $\Gamma_j$ are
polynomials of degree at most $\Delta_j$. The same 
holds for the product $\omega_{j+1}^{k+\ell}\gamma_1\gamma_2$ if
$k+\ell<p$. If $k+\ell>p$, then $\omega_{j+1}^{k+\ell}=
\omega_{j+1}^p\omega_{j+1}^{k+\ell-p}=(\omega_{j+1}+\alpha_j)
\omega_{j+1}^{k+\ell-p}$,
whence $\omega_{j+1}^{k+\ell}\gamma_1\gamma_2$
is the sum of $\omega_{j+1}^{1+k+\ell-p}\gamma_1\gamma_2$
and $\alpha_j\gamma_1\gamma_2$. The former term has coordinates
of degree at most $\Delta_j$, the coordinates of the latter 
are polynomials of degree at most $(2p+1)\Delta_j+\Delta_j+\Delta_j=
(2p+3)\Delta_j$.

Now consider the conjugate of $\omega_{j+1}^k\gamma$ 
by $\sigma_{j+1}^\ell$, where $1\leq \ell<p^{j+1}$, 
$1\leq k\leq p-1$ and $\gamma\in \Gamma_j$.
This conjugate is 
$(\omega_{j+1}^{\sigma_{j+1}^\ell})^k\gamma^{\sigma_{j+1}^\ell}$.
The second term
equals $\gamma^{\sigma_j^\ell}$ which has coordinates of degree
at most $\Delta_j$. To investigate the first term,
recall that $\omega_{j+1}^{\sigma_{j+1}}=\omega_{j+1}+\beta_j,$
whence
$$\omega_{j+1}^{\sigma_{j+1}^{\ell}}=
\omega_{j+1}+\sum_{r=0}^{\ell-1}
\beta_j^{\sigma_j^r}$$
The element $\delta=\sum_{r=0}^{\ell-1}
\beta_j^{\sigma_j^r}$, expressed in terms of $\Gamma_j$,
has again polynomial coordinates of degree at most $\Delta_j$.
Then $(\omega_{j+1}^{\sigma_{j+1}^\ell})^k$
is the sum (with binomial coefficients) of terms 
of the form $\omega_{j+1}^r\delta^{k-r}$.
The power $\delta^{k-r}$ has coordinates 
of degree at most $(k-r)\Delta_j+(k-r-1)\Delta_j\leq (2p-1)\Delta_j$
in terms of $\Gamma_j$, whence we conclude
that $(\omega_{j+1}^{\sigma_{j+1}^\ell})^k$
has, in terms of $\Gamma_{j+1}$ polynomial coordinates 
of degree at most $(2p-1)\Delta_j$. It follows that the matrix
of any power of $\sigma_{j+1}$ has polynomial entries
of degree at most $2p\Delta_j$.

We obtained that the function $(2p+3)^s=\poly(p^s)$ 
is an upper bound for both the structure constants and for 
the matrices of the powers of $\sigma_s$.
\end{proof}

\begin{lemma}
\label{lem:cyclic_gen}
Let $\F'$ be a field. Let $d$ be any non-negative integer. If $\fdchar(\F')=0$ 
then $d_1=d$. If $\fdchar(\F')=p>0$ then let $d_1$ be the $p$-free part of $d$, 
that is, $d=d_1p^s$, where $p\nmid d_1$ and $s\in \N$.
%Let $p$ be a prime such that $d=d_1p^s$ where 
%$d_1$ is the largest divisor of $d$ not divisible by $p$.
%In particular $p$ is an arbitrary prime and $s=0$ if the characteristic of $K$
%is zero while otherwise $p$ is the characteristic of $K$. 
Assume that $\F'$
contains a known $d_1$th root of unity $\zeta$. Then 
a cyclic extension $\ELL$ degree $d$ of $\K:=\F'(X)$ can 
be computed
using $\poly(d)$ arithmetic operations. $\ELL$ will be given by structure
constants with respect to a basis, and the matrix for a generator of
the Galois group in terms of the same basis will also be given.
All the output entries (the structure constants as well as the entries 
of the matrix representing the Galois group generator) will be
polynomials of degree $\poly(d)$ in $\F'[X]$.
Furthermore for $\F'=\Q[\sqrt[d_1]{1}]$, the bit complexity 
of the algorithm (as well as the size of the output) is
$\poly(d)$. 
\end{lemma}

\begin{proof}
Put $\ELL_1=\F'(Y)$ and $X=Y_1^{d_1}$.
Then $1,Y_1,\ldots,Y_1^{d_1}$ are a $\F'(X)$-basis for
$\ELL_1$ with $Y_1^iY_1^j=Y_1^{i+j}$ if $i+j\leq d_1$ and
$XY_1^{i+j-d_1}$ otherwise. Further note that the linear extension $\sigma_1$ of 
the map
sending $Y_1^j$ to $\zeta^jY_1^j$ is an automorphism of degree
$d_1$. Then $\ELL_1$ is a cyclic extension of $\F'(X)$ of degree $d_1$. This 
procedure has been used in \cite{IQS1}. 

We can compute whether $\fdchar(\F')$ is a divisor of $d$
% of $\F'$
by testing the multiples of the identity element up to $d$.
%\ynote{What if the field 
%characteristic turns out to be very large?} 
If $\fdchar(\F')=0$, or if 
$\fdchar(\F')=p>0$ and $p\nmid d$, we are done. Note that in the following $p\leq 
d$. 

If $\fdchar(\F')=p>0$ and $p\mid d$, let $d_1$ be in the statement, so $d=d_1p^s$. 
Let $d_2=p^s$, and $\F_p$ be the prime field of $\F'$. Construct the cyclic 
extension 
of degree $d_2$ of $\F_p(X)$ over $\F_p$ by ~\ref{lem:cyclic_p}, 
%tensor it with $\F'$, 
and let the resulting field be $\ELL_2$. We also
obtain the matrix a generator $\sigma_2$
of the Galois group. Then put 
$\ELL=\ELL_1\otimes_{\F_p(X)} \ELL_2$. It contains a copy of
$\K=\F'(X)\cong \F'(X)\otimes_{\F_p(X)}\F_p(X)$. We take
the product basis for the structure constants
and for matrix representation of 
the automorphism $\sigma_1\otimes \sigma_2$.
\end{proof}

\section{The complete constructive regularity lemma}\label{sec:reg} 

We first present the formal statement of the regularity lemma in its full 
generality. We also add a technical notion that will be useful for the proof 
of Theorem~\ref{thm:main}. Let $n\in \N$, and let
$\vec{i}=(i_1, \dots, i_r)$, $\vec{j}=(j_1, \dots, 
j_r)$ be two sequences of integers, where $1\leq i_1<\dots <i_r\leq n$ 
and $1\leq j_1 < \dots < j_r\leq n$. For a matrix $A\in M(n, \F)\otimes M(d, 
\F)$, the $r\times r$ \emph{window} indexed by $\vec{i}$, 
$\vec{j}$ is the sub-matrix of $A$ consisting of the blocks indexed by $(i_k, 
j_\ell)$, $k, \ell\in[r]$.

\begin{lemma}[{Regularity of blow-ups}]%\cite[Lemma 11]{IQS1}
\label{lem_reg_blowup-window}
For $\cB\leq M(n, \F)$ and $\cA=\rblowup{\cB}{d}$, assume that 
$|\F| = (rd)^{\Omega(1)}$. %(which is obviously true when $\mathrm{char}(\F)=0$). 
Given a matrix 
$A\in 
\cA$ with $\rk A > (r-1)d$, there exists a deterministic algorithm 
that returns 
$\widetilde{A}\in \cA$ and an $r\times r$ window $W$ in $\widetilde{A}$ such that 
$W$ 
is nonsingular (of rank $rd$). 
This algorithm uses $\poly(nd)$ arithmetic 
operations and, over $\Q$, the algorithm runs in polynomial time.
In particular, all intermediate numbers have bit lengths polynomial in 
the input size.
\end{lemma}

The cases (a) $\mathrm{char}(\F)=0$, (b) $\mathrm{char}(\F)$ and $d$ are coprime, 
and $|\F| = (rd)^{\Omega(1)}$ were settled
in \cite[Lemma 5.7]{IQS1} which was reproduced in  ~\ref{step2}. The main issue
with the case when $d$ is not coprime to $\mathrm{char}(\F)$ was
that we did not have an efficient construction
of an appropriate Artin-Schreier-Witt extension of $\F_p(x)$,
Now we have such a construction in Lemma~\ref{lem:cyclic_p}. 
% which gives an 
%efficient construction of an 
%Artin--Schreier--Witt extension of $\F_p(x)$.

The proof makes use of the following two results from \cite{IQS1}.

\begin{proposition}[\mbox{\cite[Proposition 4.4]{IQS1}}]
\label{prop:cyclic_algebra}
Let $\ELL$ be a cyclic extension of degree $d$ of a field $\K$, and suppose that 
$\ELL$ is given by structure constants w.r.t. a $\K$-basis $A_1, \ldots, A_d$. 
Similarly, a generator $\sigma$ for the Galois group is assumed to be given
by its matrix in terms of the same basis.
Let 
$Y$ be 
a formal variable.
Then one can construct a $\K(Y)$-basis
$\Gamma$ of $M(d,\K(Y))$  such that
the $\K(Y^d)$-linear span of $\Gamma$
is a central division algebra over $\K(Y^d)$
of index $d$, using $\poly(d)$ arithmetic operations in $\K$.
Furthermore for $\K=\Q[\sqrt[d]{1}]$, the bit complexity 
of the algorithm (as well as the size of the output) is
also $\poly(d)$. 
\end{proposition}

\begin{lemma}[\mbox{Conditional regularity~\cite[Lemma~5.4]{IQS1}}]
\label{lem:reg_technical}
Assume that we are given a matrix 
$A\in \rblowup{\cB}{d}\leq M(dn, \F)$ with $\rk(A)=(r-1)d+k$ for some
$1<k<d$. Let $X$ and $Y$ be formal variables and put
$\K=\F'(X)$, where $\F'$ is a finite extension of $\F$ of degree at most $d$. 
Suppose further that $|\F|>(nd)^{O(1)}$ and that
we are also given a $\K(Y)$-basis $\Gamma$ of $M(d,\K(Y))$  such that
the $\K(Y^d)$-linear span of $\Gamma$
is a central division algebra $D'$ over $\K(Y^d)$. Let $\delta$
be the maximum of the degrees of the polynomials appearing as
numerators or denominators of the entries of the matrices 
in $\Gamma$. Then, using $(nd+\delta)^{O(1)}$ arithmetic
operations in $\F$, one can find a matrix $A''\in \rblowup{\cB}{d}$ with 
$\rk(A'')\geq rd$. Furthermore, over $\Q$ the bit complexity of
the algorithm is polynomial in the size of the input data (that is,
the total number of bits describing the entries of matrices and in the
coefficients of polynomials).
\end{lemma}

\begin{proof}[Proof of Lemma~\ref{lem_reg_blowup-window}]
The statement, except the window part, readily follows by plugging 
Lemma~\ref{lem:cyclic_gen} of the previous section
to Proposition~\ref{prop:cyclic_algebra} and the using that in Lemma~\ref{lem:reg_technical}.
%\cite[Lemma 5.4]{IQS1}\ynote{Need to adjust 
%correspondingly 
%later. Or we can state that lemma here.}. 
To see that such a window can be computed, we first observe that the lemma 
applies to $d$-blow-ups of rectangular matrices, by simple zero padding.
Second, apply the lemma and find an $rd\times rd$ nonsingular
sub-matrix of the given matrix $A$. If the column indices include some such that 
not all of
its $d-1$ siblings are included, then (1) delete the corresponding column
from the original matrix space; (2) let $A'$ be the matrix obtained by deleting 
the corresponding $d$ columns from $A$. Then $\rk(A')>\rk(A)-(d-1)$. So we apply 
the regularity lemma in the rectangular space with $A'$, to round up the rank to 
$\rk(A)$ again. Do the same for row indices. Iterate until we obtain a 
full window.
\end{proof}

\section{Proof of the main theorem}\label{sec:main}

%We also explain why it runs in exponential time and 
%how the idea in the proof of Theorem~\ref{thm:deg_bd} can reduce its complexity 
%to polynomial. 
In Section~\ref{subsec:new_algo} we prove
Theorem~\ref{thm:main}, and in Section~\ref{sec:small} we deal with the small field case. The main drawback of our earlier algorithm discussed in
Section~\ref{step3} was that the  blow-up size increases exponentially. 
%\kvnote{changed}
However, a simple reduction procedure as described in Lemma~\ref{lem:greedy} below
readily implies that, 
once we find $A'$ of rank $r'd$ in $\rblowup{\cB}{d}$, we can efficiently reduce 
$d$ to be no 
more than $r'+1$. This means that we can always ensure that the blow-up factor is 
small, which is the key to reducing the complexity of the algorithm from 
exponential time  to 
polynomial time. We shall make the above idea rigorous in the next subsection. 

\subsection{The algorithm for the main theorem}\label{subsec:new_algo}

%In this subsection we prove Theorem~\ref{thm:main}. Here it is 
%easier to work 
%with $\rblowup{\cB}{d}:=\cB\otimes M(d, \F)$ instead of $\rblowup{\cB}{d}=M(d, 
%\F)\otimes \cB$. This does not change anything, as $\rblowup{\cB}{d}$ is 
%naturally 
%isomorphic to $\rblowup{\cB}{d}$. The only point is that we shall think of 
%matrices in the blow-up space as 
%$n\times n$ block matrices with each block of size $d\times d$. 
We first recall 
some 
preparation material from \cite{IQS1}. 

%\kvnote{changed mildly}
Finding an $sd$-shrunk subspace for the $\rblowup{\cB}{d}$
is equivalent to finding an $s$-shrunk subspace for $\cB$ 
because of the following simple observations (\cite[Proposition 5.2]{IQS1}). 
Firstly, for every $s$-shrunk subspace $U$ of $\F^n$ the
subspace $U\otimes \F^d$ for $\cB$
is an $sd$-shrunk
subspace for $\rblowup{\cB}{d}$. Conversely, a
$s'$-shrunk subspace for $\rblowup{\cB}{d}$ can be
embedded into a subspace of the form $U\otimes \F^d$
where $U$ is an $s$-shrunk subspace for $\cB$
with $sd\geq s'$.
%
%The next lemma concerns
%a simple procedure originally developed in \cite{GIR} 
%for reducing data for an algorithm for a task in a different,
%although not completely unrelated
%context. See also
%the appropriate paragraph
%in Subsection~4.1 of \cite{IQS1} for a brief description.
%This technique 
%is extensively used in the above mentioned two results for 
%handling extension fields if necessary. It is also used
%to keep data small over infinite base fields. Here it can also 
%be applied to keep the entries of the $d$ by $d$ matrices defining the large
%rank matrix from the $d$-blow-up in a finite (but sufficiently large)
%subset of the base fields. 

The main technical ingredient of our algorithm is an improvement of \cite[Theorem 
5.10]{IQS1}, discussed in  Section~\ref{step2} . 
It states that either
a shrunk subspace witnessing that the (scaled-down) rank of a 
matrix in a blow-up reaches the non-commutative rank
or a matrix in a larger blow-up having 
larger scaled-down rank can be efficiently constructed. For completeness we give all the details
and also the proof even though it is identical to that in our earlier paper 
excepting for the last step.
\begin{theorem}
\label{thm-blup-incr}
Let $\cB\leq M(n, \F)$ and let $\cA=\rblowup{\cB}{d}$. Assume that
we are given a matrix $A\in \cA$ with $\rk(A)=rd$, and
$|\F|$ is $(ndd')^{\Omega(1)}$, where $d'=r+1$. There 
exists a deterministic algorithm that returns either an $(n-r)d$-shrunk 
subspace for $\cA$ (equivalently, an $(n-r)$-shrunk subspace for $\cB$), or a 
matrix $B\in \cA\otimes M(d', \F)$ of rank at least $(r+1)dd'$. Furthermore, in 
the latter case an $(r+1)\times (r+1)$ window is also found
such that the corresponding $(r+1)dd'\times (r+1)dd'$ 
sub-matrix of $B$ has full rank. This algorithm 
uses $\poly(ndd')$ arithmetic operations and, over $\Q$, all intermediate numbers 
have bit lengths polynomial in the input size.
\end{theorem}
\begin{proof}
Starting with the kernel $V_0$ of the linear map $A$ we compute the image $W_1$ of $V_0$ under $\cA$.
If $W_1$ is not in the image of $A$ we stop and declare $W^{*}=W_1$. Otherwise we define $V_1$ to be the preimage of $W_1$ under $A$ and define 
$W_2$ to be the image of $V_1$ under $\cA$.  We continue doing so, at each step checking if $W_i$ is in the image of
$A$ or not. Since at each step the dimension of $W_i$ increases by $d$ it is clear that we halt in $l$ steps with $l$ at most $r+1$, obtaining the limit subspace $W^{*}=W_l$. If $W^{*}$ is in the image of $A$,
 it follows from Fact~\ref{fact:Wong} that the preimage of $W_l$ under $A$ is an 
 $(n-r)d$-shrunk subspace. In either case in at most $r+1$ steps we find a shrunk 
 subspace or find that $W^{*}$ is not in
the image of $A$.

When the limit subspace is not in $\im(A)$ we proceed as follows. Let $B_l$ be an element of $\cA$ and $v_l \in V_{l-1}$ such that $B_l(v_l) \not \in \im(A)$. Then find matrices $B_{l-1} \in \cA$ and vector $v_{l-1} \in V_{l-2}$ such
that $B_{l-1}(v_{l-1}) = A(v_l)$. Walking backwards, we find matrices $B_{l-2},\ldots, B_1$ and vectors $v_{l-3}, \ldots, v_1$, $v_i \in V_{i-1}$ such that $A(v_i)=B_{i-1}(v_{i-1})$. In particular $v_1 \in \ker(A)$.

Now let $A' = A \otimes I_{d'}$. Clearly $A'$ is a matrix of rank $rdd'$ in $\cA^{d'}=\cB^{dd'}$.  Now let $E_{i, j}$ be the elementary matrix in $M(d',\F)$ with the $(i,j)$th entry being $1$ and
others $0$. Put $\widehat{B}=B_1\otimes E_{1,2}+B_2\otimes E_{2,3}+\ldots                                   
B_{l-1}\otimes E_{l-1,l}+B_l\otimes E_{l,1} \in \rblowup{\cB}{dd'}$.  If the rank of $\widehat{B} $ is more than $rdd'$ we set $A''$ to be $\widehat{B}$.  Otherwise consider the vectors 
$w_1=v_1 \otimes u_1$,  $w_2=v_2 \otimes u_2$,  $\ldots$ , $w_l=v_l \otimes u_l$. It is clear that 
 $A'(w_1)=0$ and that $A'w_{j} = \widehat{B}(w_{j-1})$ for $2 \leq j \leq l$. Furthermore, $\widehat{B}(w_l) = B_l(v_l) \otimes u_{l+1}$ and this is not in $A'(\F^{nd} \otimes \F^{d'})$ since $B_l(v_l)$ is not in the image of $A$.
So if we were to compute the second Wong sequence starting with the matrix $A'$ in the rank two linear space of $\cB^{dd'}$ spanned by matrices $\{A', \widehat{B}\}$,  the second Wong sequence 
runs out of the image of $A'$. So by Fact~\ref{fact:dimtwo} $A'$ is not of maximal rank in the linear space spanned by $\{A', \widehat{B}\}$.
So there exists $\mu \in \F$ such that 
$A' + \mu \widehat{B}$ has rank strictly bigger than $rdd'$.
%\gabor{Replaced $\lambda$ with 1 and changed the next sentence.} 
As the determinant
of an $(rdd'+1)\times (rdd'+1)$ 
submatrix of $A' + \mu \widehat{B}$ is
a polynomial of degree at most $rdd'+1$ in $\mu$,  
we can find $\mu$ by 
running over all of elements of a subset of
$\F$ of size $rdd'+2$
till we find one.

We then invoke  Lemma~\ref{lem_reg_blowup-window} with $A''$ to obtain a matrix 
$B$ over the base field $\F$ of rank $(r+1)dd'$ and the $(r+1) \times (r+1)$ 
window as required, completing the proof.

It is clear that the matrices $B_1,\ldots,B_l$ as well as $\mu$
can be determined in the given 
polynomial time.
%\gabor{There was a too large time bound
%$|\F|$ for $\mu$ and $\lambda$.}
\end{proof}
%\kvnote{changed}
To obtain the algorithm for Theorem~\ref{thm:main}, the regularity lemma needs to 
be accompanied with a reduction procedure that keeps the blow-up parameter small. 
We mentioned in the introduction that there are two
methods for this purpose, and in this section we use our method. 
%, which just follows the idea in the proof of
%In this section we do so by using the idea in the proof of 
%Theorem~\ref{thm:deg_bd}. 
The method based on the Derksen-Makam technique is presented in 
Section~\ref{sec:dm}. 
\begin{lemma}\label{lem:greedy}
Let $\cB\leq M(n, \F)$, and $d>n+1$. Assume we are given a matrix $A\in 
\cB^{\{d\}}$ of rank $dn$. Then there exists a deterministic polynomial-time 
procedure 
that 
constructs $A'\in \cB^{\{d-1\}}$ of rank $(d-1)n$. 
\end{lemma}
\begin{proof}
Let $A''$ be an appropriate $(d-1)n\times (d-1)n$ sub-matrix of $A$ corresponding 
to a 
matrix in $\cB^{\{d-1\}}$. We claim $A''$ is of rank $>(d-1)(n-1)$. Suppose 
not, as $A$ is obtained from $A''$ from adding $n$ rows and then $n$ columns, and 
$d>n+1$, we have $\rk(A)\leq \rk(A'')+2n\leq dn-d-n+1+2n < dn$, a contradiction. 
Now that $\rk(A'')>(d-1)(n-1)$, using Lemma~\ref{lem_reg_blowup-window}, we obtain 
$A'\leq \cB^{\{d-1\}}$ of rank $(d-1)n$. 
\end{proof}

\begin{proof}[Proof of Theorem~\ref{thm:main}]
Let $B_1,\ldots,B_m$ be the input basis for $\cB$.
The algorithm is an iteration based on 
Theorem~\ref{thm-blup-incr}. In each round we start with a matrix $A=\sum_i 
B_i\otimes T_i\in \rblowup{\cB}{d}$
of rank $rd$  for some integer $d\leq r+1$. 
%We assume that
%at least one of the basis elements $B_i$ has rank larger
%than 1, as otherwise \cite{IKS} works even with $d=1$.
%Initially $d=1$ and $A$ is a basis element having
%rank at least 2. 
In the first round, $d=1$ and $A$ can be taken as any matrix in $\cB$.
The procedure behind 
Theorem~\ref{thm-blup-incr} either returns an
$(n-r)$-shrunk subspace (in which case we are done),
or a new matrix (denoted also by $A$)
in a blow-up $\rblowup{\cB}{d'}$ of rank $\geq (r+1)d'$
for some $d'\leq (r+1)^2$, together with a square window of size
$r+1$ so that the corresponding sub-matrix of $A$ 
is of rank $(r+1)d'$. 
%Put $r'=r$ when
%$\mathrm{char}(\F)$ is zero or larger than
%$2(r+1)$ and put $r'=2r+1$ otherwise.
If 
%$d'>r'$ 
$d'>r+2$
we apply Lemma~\ref{lem:greedy} as follows. The $n$ in the statement of 
Lemma~\ref{lem:greedy} will be $r+1$, 
and we use it repeatedly to get a matrix in the $(r+2)$-blow-up, the 
main content
of which consists of 
%$r'$ by $r'$
$(r+2)\times(r+2)$
matrices $T_1',\ldots,T_m'$ such that the corresponding
%$(r+1)r'\times (r+1)r'$
$(r+1)(r+2)\times (r+1)(r+2)$
sub-matrix of $A'=\sum_i B_i\otimes T_i'$ has full rank. 
Then we replace $A$ with $A'$ and apply the size reduction
procedure in Lemma~\ref{lem:red_data} to arrange that the entries 
of $T_i$ fall into the prescribed subset of $\F$,
and continue the iteration with this new matrix $A$.
\end{proof}

\subsection{Proof of Corollary~\ref{cor:small}: the case of small finite 
fields}\label{sec:small}

%\ynote{Added this section:} 
We only need to prove Corollary~\ref{cor:small} (2), 
from which (1) and (3) are 
immediate.

Given a matrix space $\cB\leq M(n, \F)$ and a field extension $\K/\F$, $\cB$ can 
be viewed naturally as a matrix space in $M(n, \K)$. 
For convenience we use $\ncrk_\F(\cB)$ to signal that we consider the 
non-commutative rank of $\cB$ over $\F$. 
%For convenience we use 
%$\cB_\F$ to signal that we consider $\cB$ as over $\F$. 
We first observe that the 
non-commutative 
rank does not change under field 
extensions. This is classical, and can be seen from the perspective of the second 
Wong sequences (see e.g. \cite[Section 2]{conf_version}). Note that 
the commutative rank may get larger if we go to an 
extension field from a too-small field. 
\begin{lemma}\label{lem:fd_ext}
Given $\cB\leq M(n, \F)$ and a field extension $\K/\F$, we have 
$\ncrk_\F(\cB)=\ncrk_\K(\cB)$.
\end{lemma}

Suppose $\cB\leq M(n, \F)$ is given by a linear basis $\{B_1, \dots, B_m\}$. Let 
$\K/\F$ be a field extension of degree $g$ so that 
$|\K|=n^{\Omega(1)}$ satisfies the field size condition of Theorem~\ref{thm:main}. 
Note that $g\leq O(\log_{|\F|} n)$.
Viewing $\cB$ as a matrix space over $\K$,  we apply Theorem~\ref{thm:main} to
compute $\ncrk_\K(\cB)$, which is equal to $r=\ncrk_\F(\cB)$ by 
Lemma~\ref{lem:fd_ext}. We also obtain the following: (1) $A_1, \dots, 
A_m\in M(d, 
\K)$ such that $A=\sum_{i\in[m]} A_i\otimes B_i$ is of rank $rd$, and (2) 
$U\leq \K^n$ such that $U$ is a shrunk subspace of $\cB$ a matrix space in $M(n, 
\K)$. We fix an embedding $\phi$ of $\K$ into $M(g, \F)$ using the regular 
representation. For $i\in[m]$, construct $\widetilde{A_i}\in M(gd, \F)$ by 
replacing each entry $\alpha$ of $A_i$ with $\phi(\alpha)$, and form 
$\widetilde{A}=\sum_{i\in[m]} \widetilde{A_i}\otimes B_i$. Note that 
$\widetilde{A}$ is in $M(gd, \F)\otimes \cB$, and it can be seen easily that 
$\rk(\widetilde{A})=g\cdot \rk(A)$. Since $\rk(\widetilde{A})/gd=r=\ncrk_\F(\cB)$, 
we 
have $\crk_\F(M(gd, \F)\otimes \cB)=\ncrk_\F(M(gd, \F)\otimes \cB)$. This implies 
that 
we can apply the second Wong sequence to $(\widetilde{A}, M(gd, \F)\otimes \cB)$ 
to obtain an $(n-r)gd$-shrunk subspace of $M(gd, \F)\otimes \cB$ which then 
induces an $(n-r)$-shrunk subspace of $\cB$.

\section{Constructivizing the result of Derksen and Makam}\label{sec:dm}

Here is an algorithmic version of Lemma 2.7 of \cite{DM2}. Although the most relevant blow-ups 
in the context of the non-commutative rank problem are 
square (e.g,  of the form $\rblowup{\cB}{k}$, described earlier), non-square blow-ups turned out
to be crucial in the reduction techniques in \cite{DM2}.
So we use a different notation for blow-ups from what was used so far. Given a matrix space $\cB\leq M(n,\F)$, its {\em $(k, \ell)$-blow-up} 
$\rblowup{\cB}{k,\ell}$
is defined as the matrix space $\cB\otimes M(k\times \ell,\F)$
in $M(nk\times n\ell,\F)$. 

\begin{lemma}
\label{lem:two.seven}
Let $\cB\leq M(n, \F)$. 
Assume that for $k,\ell=1,\ldots,N$ we are given
matrices $M_0(k,\ell)\in \rblowup{\cB}{k,\ell}$ of rank $r_0(k,\ell)$, 
and suppose that 
$|\F|\geq 2nN+1$.
Then for every $k,\ell=0,\ldots,N$ we can efficiently 
(that is, by an algorithm that
uses $\poly(Nn)$ arithmetic operations and, over e.g.~$\Q$,
produces intermediate and final data of size polynomial in the
input size) construct
matrices $M(k,\ell)\in \rblowup{\cB}{k,\ell}$ of rank 
$r(k,\ell)\geq r_0(k,\ell)$ such
that
\begin{enumerate}
%\item[(0)] $r(k,\ell)\geq r_0(k,\ell)$ 
\item[(1)] $r(k,\ell+1)\geq r(k,\ell)$ ($0\leq \ell<N$);
\item[(2)] $r(k+1,\ell)\geq r(k,\ell)$ ($0\leq k<N$);
\item[(3)] $r(k,\ell+1)\geq \frac{1}{2}(r(k,\ell)+r(k,\ell+2)$ ($0\leq \ell<N-1$);
\item[(4)] $r(k+1,\ell)\geq \frac{1}{2}(r(k,\ell)+r(k+2,\ell)$ ($0\leq k<N-1$);
\item[(5)] 
%if $k$ is not divisible by
%$\mathrm{char}(\F)$ then 
$r(k,k)$ is divisible by $k$.
\end{enumerate}
For $k=0$ (resp. $\ell=0$) we assume that $M_0(k,\ell)$ is 
the empty matrix having $\ell$ 
columns (resp. $k$ rows), and $r(k, \ell)=0$.
%having $k$ rows and $\ell$ columns.
\end{lemma}

\begin{proof}
Initially put $M(k,\ell)=M_0(k,\ell)$ for every pair $(k,\ell)$.
For a $k\times \ell$ matrix $T$ let $T^{+}$ denote the $(k+1)\times \ell$
matrix obtained form $T$ by appending a zero ($(k+1)$st) row, 
$T^{++}$ is obtained by appending two zero rows. 
%For a $(p+1)\times q$
%matrix $T$ we denote by $T^{^}$ the $(p+2)\times q$ matrix obtained
%form $T$ by inserting a zero row as $(p+1)$st.
For $M=\sum_{i=1}^m B_i\otimes T_i$ we use $M^{+}$ for
$\sum_{i=1}^m B_i\otimes T_i^{+}$, while
$M^{++}=\sum_{i=1}^m B_i\otimes T_i^{++}$.
%and
%$M^{^}=\sum_{i=1}^m B_i\otimes T_i^{^}$.

Let $(k,\ell)$ be a pair
such that any of (1)--(5) is violated. Then we will
replace some of the matrices $M(k',\ell')$ with
matrices having larger rank. Over an infinite base field like $\Q$,
each  such replacement step (or each small group consisting 
of a few them) can be followed by an application of the data reduction
procedure in Lemma~\ref{lem:red_data}  to keep intermediate (as well as the final)
data small.

If (1) is violated
then, like in \cite{DM2}, replace $M(k+1,\ell)$ with $M(k,\ell)^{+}$.
We can treat a violation of (2) symmetrically. 

When (3) is violated we consider the matrix
$A=A(t)=M(k+2,\ell)+tM(k,\ell)^{++}$ as a $(k+2)\times \ell$ block matrix
consisting of square blocks of size $n$ from $\cB$. We can
choose $t$ from any subset $S$ of size $2nN+1$ of the base
field so that $A$ has rank at least $r(k+2,\ell)$,
while the first $kn$ rows form a matrix of rank at least
$r(k,\ell)$. This is because
a necessary condition for violating either of these two conditions
is that the determinant of an appropriate (but unknown) 
sub-matrix vanishes which determinant is, as a polynomial
of degree at most $nN$ in $t$ is not identically zero.
The product of these polynomials has degree at most $2nN$
therefore it cannot have more that $2nN$ zeros. 

If $A$ has rank larger than $r(k+2,\ell)$ then
we replace $M(k+2,\ell)$ with $A$. Otherwise,
like in \cite{DM2}, let $U$ be the span of the first 
$kn$ rows of $A$, $V$ be the span of the first $(k+1)n$ rows and
$W$ be the span of the first $kn$ rows and the last $n$ rows.
Note that these collections rows correspond to matrices
of the form $A_0=\sum B_i\otimes T_i$, $A_1=\sum B_i\otimes T_i'$
and $A_2=\sum B_i\otimes T_i''$ where $T_i$ are $k\times \ell$
matrices, while $T_i'$ and $T_i''$ have  $(k+1)$ rows and $\ell$ columns. 
As $U\leq V\cap W$ and the row
space of $A$ is $V+W$, we have 
$r(k,\ell)\leq \dim U\leq \dim (V\cap W)=\dim V+\dim W-\dim V+W=
\dim V+\dim W-r(k+2,\ell)$. It follows that $\dim V+\dim W\geq
r(k,\ell)+r(k+2,\ell)$, whence violation of (3) is only possible
if either $\dim V$ or $\dim W$ is strictly larger than
$\frac{1}{2}(r(k,\ell)+r(k+2,\ell))$. Then we replace $M(k+1,\ell)$ with
$A_1$ or $A_2$, according to which one has larger rank.
A violation of $(4)$ is treated symmetrically.

When (5) is violated then we can apply~\ref{lem_reg_blowup-window}.
%let $d=\mbox{gcd}(k,\ell)$, $k=k'd$ $\ell=\ell'd$.
%W.l.o.g. we may  assume that $k'\leq \ell'$ and let $\cB'$ be the space
%consisting of the of $\ell'\times \ell'$ matrices whose 
%last $\ell'-k'$ rows are zero. Then we can apply Lemma~\ref{lem_reg_blowup} to 
%the 
%blow-up $\rblowup{(\cB\otimes \cB')}{d,d}$.

As in each round when violation of (1),$\ldots$,(4) or (5)
 occurs the rank of at least one 
of the matrices $M(k,\ell)$ is incremented, the total number of  
rounds for achieving (1)--(5) is at most $N^3n$.
\end{proof}

And here is essentially Proposition 2.10 of \cite{DM2}. 
We include a proof (which is almost literally the same as
the proof in \cite{DM2})
here for completeness. We note that this lemma deals only with the property of 
certain families of functions, without referring to matrices.

\begin{lemma}[{\cite[Proposition 2.10]{DM2}}]
\label{lem:two.ten}
Assume that $N>n>0$, $r:\{0, 1, \ldots, N\}^2\rightarrow \Z$ is a function
with $0\leq r(k,\ell)\leq \min(k,\ell)n$ for $k, \ell \in \{0, 1, \dots, N\}$
also satisfying (1)--(5) of
~\ref{lem:two.seven}. Suppose further that $r(1,1)>1$, 
%$r(k, 0)=r(0, \ell)=0$ for $k, \ell \in \{0, 1, \dots, N\}$, 
%the above follows form leq min(k,\ell)n 
and there exists
$d$ such that $n\leq d+1\leq N$ and $r(d+1,d+1)=n(d+1)$.
Then, 
%if $d$ is not divisible by $\mathrm{char}(\F)$ then 
$r(d,d)=nd$ as well.
%Otherwise, if $d\geq 2n$ then we have $r(d-1,d-1)=n(d-1)$.
\end{lemma}
\begin{proof}
By $r(d+1, d+1)=n(d+1)$, for $1\leq a<d+1$,
$$
r(d+1, a)\geq \frac{(d+1-1)\cdot r(d+1, 0) + a\cdot r(d+1, d+1)}{d+1}=an.
$$
As by assumption $r(d+1, a)\leq an$, we have $r(d+1, a)=an$. Similarly $r(a, 
d+1)=an$ for $1\leq a<d+1$.

Then we bound $r(1, d)$ as follows:
\begin{eqnarray*}
r(1, d) & \geq & \frac{(d-1)\cdot r(1, d+1)+1\cdot r(1, 1)}{d} \\
 & \geq & \frac{(d-1)n+2}{d}=n-\frac{n-2}{d}>n-1. 
\end{eqnarray*}
Note that we use $r(1, 1) > 1$ and $d\geq n-1$. Since $r(1, d)\in \Z$, $r(1, d)=n$.

We are ready to bound $r(d, d)$ then. 
\begin{eqnarray*}
r(d, d) & \geq & \frac{(d-1)\cdot r(d+1, d)+1\cdot r(1, d)}{d} \\
 & = & \frac{(d-1)dn+n}{d} = nd-n+\frac{n}{d}.
\end{eqnarray*}
{}From $d\geq n-1$ it is inferred easily that $-n+\frac{n}{d}>-d$. Therefore 
$nd-n+\frac{n}{d}>(n-1)d$. By (5) we conclude that $r(d, d)=nd$.
\end{proof}

We finally remark that, if we use Lemma~\ref{lem:two.seven} in the proof of 
Theorem~\ref{thm:main}, then 
$n$ in the statement of the lemma will be $r+1$, $N$ will be 
$d'$, $M_0(d',d')$ is the nonsingular $(r+1)d'\times (r+1)d'$
block of $A$ and $M_0(p,q)$ can be actually even the zero matrix
for $(p,q)\neq (d',d')$.
It will prepare matrices in several not necessarily square blow-ups, 
among others, most importantly, one in an 
%$(r',r')$-blow-up.
$(r,r)$-blow-up with a similar content as described in the proof of 
Theorem~\ref{thm:main}.
%\begin{acknowledge} 

\paragraph{Acknowledgements.}
We would like to thank the authors
of \cite{GGOW} and of \cite{DM2} for sharing their ideas with us
and making us possible to read early versions of their manuscripts.
Part of the work was done when G\'abor and Youming 
were visiting the Centre for Quantum Technologies at
the National University of Singapore.
Research of the first author was also supported in part by the
Hungarian National Research, Development and Innovation Office – NKFIH
Grant 115288.
Youming's 
research was supported by the Australian Research Council DECRA DE150100720. 
KV's research was supported by a grant from the Infosys foundation.
%\end{acknowledge}

%\bibliographystyle{alpha}
%\bibliography{references_new}

\begin{thebibliography}{GGOW16}

\bibitem[AL81]{PrimitiveI}
MD~Atkinson and S~Lloyd.
\newblock Primitive spaces of matrices of bounded rank.
\newblock {\em Journal of the Australian Mathematical Society (Series A)},
  30(04):473--482, 1981.

\bibitem[Ami66]{Ami66}
S.A Amitsur.
\newblock Rational identities and applications to algebra and geometry.
\newblock {\em Journal of Algebra}, 3(3):304 -- 359, 1966.

\bibitem[ANS07]{ANK07}
B.~Adsul, S.~Nayak, and K.~V. Subrahmanyam.
\newblock A geometric approach to the {Kronecker problem II: rectangular
  shapes, invariants of matrices and the Artin{\textendash}Procesi theorem}.
\newblock preprint, 2007.

\bibitem[BD06]{BD06}
M.~B{\"u}rgin and J.~Draisma.
\newblock The {Hilbert} null-cone on tuples of matrices and bilinear forms.
\newblock {\em Mathematische Zeitschrift}, 254(4):785--809, 2006.

\bibitem[Ber70]{Berg70}
George~W. Bergman.
\newblock Skew fields of noncommutative rational functions (preliminary
  version).
\newblock {\em S\'eminaire Sch\"utzenberger}, 1:1--18, 1969-1970.

\bibitem[BFS99]{BFS}
Jonathan~F. Buss, Gudmund~S. Frandsen, and Jeffrey~O. Shallit.
\newblock The computational complexity of some problems of linear algebra.
\newblock {\em J. Comput. Syst. Sci.}, 58(3):572--596, 1999.

\bibitem[CIKK15]{CIKK15}
Marco Carmosino, Russell Impagliazzo, Valentine Kabanets, and Antonina
  Kolokolova.
\newblock Tighter connections between derandomization and circuit lower bounds.
\newblock In {\em Approximation, Randomization, and Combinatorial Optimization.
  Algorithms and Techniques, {APPROX/RANDOM} 2015, August 24-26, 2015,
  Princeton, NJ, {USA}}, pages 645--658, 2015.

\bibitem[Coh85]{Cohn}
P.~M. Cohn.
\newblock {\em Free Rings and Their Relations}.
\newblock L.M.S. Monographs. Acad. Press, 1985.
\newblock First edition 1971.

\bibitem[Coh95]{Cohn95}
P.~M. Cohn.
\newblock {\em Skew Fields: Theory of General Division Rings}.
\newblock Encyclopedia of Mathematics and its Applications. Cambridge
  University Press, 1995.

\bibitem[CR99]{CR99}
P.~M. Cohn and C.~Reutenauer.
\newblock On the construction of the free field.
\newblock {\em International Journal of Algebra and Computation},
  9(3-4):307--323, 1999.

\bibitem[Der01]{derksen_bound}
Harm Derksen.
\newblock Polynomial bounds for rings of invariants.
\newblock {\em Proceedings of the American Mathematical Society},
  129(4):955--964, 2001.

\bibitem[dGIR96]{GIR}
Willem~A. de~Graaf, G{\'a}bor Ivanyos, and Lajos R{\'o}nyai.
\newblock {Computing Cartan subalgebras of Lie algebras}.
\newblock {\em Applicable Algebra in Engineering, Communication and Computing},
  7(5):339--349, 1996.

\bibitem[DM17a]{DM3}
H.~Derksen and V.~Makam.
\newblock On non-commutative rank and tensor rank.
\newblock {\em Linear and Multilinear Algebra}, pages 1--16, 2017.
\newblock Article in Press.

\bibitem[DM17b]{DM2}
H.~Derksen and V.~Makam.
\newblock Polynomial degree bounds for matrix semi-invariants.
\newblock {\em Advances in Mathematics}, 310:44--63, 2017.

\bibitem[DW00]{DW00}
Harm Derksen and Jerzy Weyman.
\newblock Semi-invariants of quivers and saturation for littlewood-richardson
  coefficients.
\newblock {\em Journal of the American Mathematical Society}, 13(3):467--479,
  2000.

\bibitem[DZ01]{DZ01}
M.~Domokos and A.~N. Zubkov.
\newblock Semi-invariants of quivers as determinants.
\newblock {\em Transformation groups}, 6(1):9--24, 2001.

\bibitem[Edm67]{Edm67}
Jack Edmonds.
\newblock Systems of distinct representatives and linear algebra.
\newblock {\em J. Res. Nat. Bur. Standards Sect. B}, 71:241--245, 1967.

\bibitem[EH88]{EH88}
David Eisenbud and Joe Harris.
\newblock Vector spaces of matrices of low rank.
\newblock {\em Advances in Mathematics}, 70(2):135 -- 155, 1988.

\bibitem[FR04]{FR04}
M.~Fortin and C.~Reutenauer.
\newblock Commutative/noncommutative rank of linear matrices and subspaces of
  matrices of low rank.
\newblock {\em S{\'e}minaire Lotharingien de Combinatoire}, 52:B52f, 2004.

\bibitem[GGOW16]{GGOW}
A.~Garg, L.~Gurvits, R.~Oliveira, and A.~Wigderson.
\newblock A deterministic polynomial time algorithm for non-commutative
  rational identity testing.
\newblock In {\em Proceedings - Annual IEEE Symposium on Foundations of
  Computer Science, FOCS}, pages 109--117, 2016.

\bibitem[Gur04]{Gurvits}
Leonid Gurvits.
\newblock Classical complexity and quantum entanglement.
\newblock {\em J. Comput. Syst. Sci.}, 69(3):448--484, 2004.

\bibitem[HW15]{HW15}
Pavel Hrube\v{s} and Avi Wigderson.
\newblock Non-commutative arithmetic circuits with division.
\newblock {\em Theory of Computing}, 11:357--393, 2015.

\bibitem[IKQS15]{conf_version}
G{\'{a}}bor Ivanyos, Marek Karpinski, Youming Qiao, and Miklos Santha.
\newblock Generalized wong sequences and their applications to edmonds'
  problems.
\newblock {\em J. Comput. Syst. Sci.}, 81(7):1373--1386, 2015.

\bibitem[IQS17]{IQS1}
G.~Ivanyos, Y.~Qiao, and K.V. Subrahmanyam.
\newblock Non-commutative edmonds’ problem and matrix semi-invariants.
\newblock {\em Computational Complexity}, 26(3):717--763, 2017.

\bibitem[KI04]{KI04}
Valentine Kabanets and Russell Impagliazzo.
\newblock Derandomizing polynomial identity tests means proving circuit lower
  bounds.
\newblock {\em Computational Complexity}, 13(1-2):1--46, 2004.

\bibitem[LSW00]{LSW00}
Nathan Linial, Alex Samorodnitsky, and Avi Wigderson.
\newblock A deterministic strongly polynomial algorithm for matrix scaling and
  approximate permanents.
\newblock {\em Combinatorica}, 20(4):545--568, 2000.

\bibitem[Mal78]{Mal78}
Peter Malcolmson.
\newblock A prime matrix ideal yields a skew field.
\newblock {\em Journal of the London Mathematical Society}, s2-18(2):221--233,
  1978.

\bibitem[Ram54]{Ramanathan}
K.G. Ramanathan.
\newblock {\em Lectures on the Algebraic Theory of Fields}.
\newblock Tata Institute of Fundamental Research, Bombay, 1954.

\bibitem[SVdB01]{SV01}
Aidan Schofield and Michel Van~den Bergh.
\newblock Semi-invariants of quivers for arbitrary dimension vectors.
\newblock {\em Indagationes Mathematicae}, 12(1):125--138, 2001.

\bibitem[Won74]{Wong}
Kai-Tak Wong.
\newblock {The eigenvalue problem $\lambda Tx+Sx$}.
\newblock {\em Journal of Differential Equations}, 16(2):270 -- 280, 1974.

\end{thebibliography}

\end{document}